\documentclass[11pt,letterpaper]{article}
\usepackage[utf8]{inputenc}
\usepackage[english]{babel}
\usepackage{amsmath}
\usepackage{amssymb}
\usepackage{fullpage}
\usepackage[colorinlistoftodos,disable]{todonotes}
\usepackage{amsthm}
\usepackage{thm-restate}
\usepackage[hidelinks]{hyperref}

\newtheorem{definition}{Definition}
\newtheorem{lemma}{Lemma}
\newtheorem{observation}{Observation}

\newtheorem{theorem}{Theorem}
\newtheorem{claim}{Claim}

\newcommand{\lef}{\texttt{left}}
\newcommand{\righ}{\texttt{right}}
\newcommand{\gap}{\texttt{gap}}
\newcommand{\num}{\texttt{num}}
\newcommand{\out}{\texttt{out}}
\newcommand{\tup}{\texttt{tup}}

\newcommand{\mmin}{\texttt{MIN}}
\newcommand{\mmax}{\texttt{MAX}}
\newcommand{\var}{\texttt{Var}}

\begin{document}

\sloppy
\author{
Alexander Kulikov\thanks{Steklov Mathematical Institute at St.~Petersburg, Russian Academy of Sciences, email: \href{mailto:kulikov@logic.pdmi.ras.ru}{kulikov@logic.pdmi.ras.ru}}
\and
Ivan Mikhailin\thanks{Department of Computer Science and Engineering, University
of California, San Diego, email: \href{mailto:imikhail@eng.ucsd.edu}{imikhail@eng.ucsd.edu}} \and
Andrey Mokhov\thanks{School of Engineering, Newcastle University, UK, email: \href{mailto:andrey.mokhov@ncl.ac.uk}{andrey.mokhov@ncl.ac.uk}}
\and
Vladimir Podolskii\thanks{Steklov Mathematical Institute, Russian Academy of Sciences and
National Research University Higher School of Economics, Moscow, email: \href{mailto:podolskii@mi-ras.ru}{podolskii@mi-ras.ru}}}
\date{}
\title{Complexity of Linear Operators}
\maketitle

\begin{abstract}
Let $A \in \{0,1\}^{n \times n}$ be a~matrix with $z$~zeroes
and $u$~ones and $x$ be an~$n$-dimensional vector of
formal variables over a~semigroup $(S, \circ)$.
How many semigroup operations are required to compute the linear operator $Ax$?

As we observe in this paper, this problem contains as a special case the well-known
range queries problem and has a~rich variety of applications in
such areas as graph algorithms, functional programming, circuit complexity,
and others. It is easy to compute $Ax$ using $O(u)$ semigroup
operations.
The main question studied in this paper is:
can $Ax$~be computed using $O(z)$ semigroup operations?
We prove that in general this is not possible: there exists
a~matrix $A \in \{0,1\}^{n \times n}$ with exactly two zeroes in every row
(hence $z=2n$) whose complexity is $\Theta(n\alpha(n))$
where $\alpha(n)$ is the inverse Ackermann function.
However, for the case when the semigroup is commutative,
we give a~constructive proof of an~$O(z)$ upper bound.
This implies that in commutative settings, complements of~sparse
matrices can
be processed as efficiently as sparse matrices (though the
corresponding
algorithms are more involved). Note that this covers the
cases of Boolean and tropical semirings that have numerous
applications, e.g., in graph theory.

As a~simple application of the presented linear-size construction,
we show
how to multiply two $n\times n$ matrices over an arbitrary
semiring in $O(n^2)$
time if one of these matrices is a~0/1-matrix with $O(n)$~zeroes
(i.e., a~complement of a~sparse matrix).
\end{abstract}

\clearpage
\tableofcontents
\clearpage

\section{Introduction}
\subsection{Problem Statement and New Results}

Let $A \in \{0,1\}^{n \times n}$ be a~matrix with $z$~zeroes
and $u$~ones, and $x=(x_1, \dotsc, x_n)$~be an~$n$-dimensional vector
of formal variables over a~semigroup~$(S, \circ)$. In this paper,
we study the
complexity of the \emph{linear operator}~$Ax$,
i.e., how many semigroup
operations are required to compute a~vector whose $i$-th element is
\[
\sum_{1 \le j \le n\,\bigwedge\,A_{ij}=1}x_j
\]
where the summation is over the semigroup operation~$\circ$.\footnote{Note that the result of
summation is undefined in case of an all-zero row, because semigroups have no
neutral element in general. One can trivially sidestep this technical issue by
adding an all-one column~$n+1$ to the matrix~$A$, as well as the neutral element
$x_{n+1}$ into the vector. Alternatively, we could switch from semigroups to
\emph{monoids}, but we choose not to do that, since we have no use
for the neutral element and associated laws in the rest of the paper.}
More specifically, we are interested in lower and
upper bounds involving~$z$ and~$u$.
Matrices with $u=O(n)$ are usually called \emph{sparse},
whereas matrices with $z=O(n)$
are called \emph{complements of sparse matrices}.
Computing all $n$~outputs
of~$Ax$ directly (i.e. using the above definition) takes
$O(u)$ semigroup operations.
The main question studied in this paper is:
can $Ax$~be computed using $O(z)$ semigroup
operations? Note that it is easy to achieve $O(z)$ complexity if $\circ$ has an
inverse. Indeed, in this case $Ax$~can be computed via subtraction:
$Ax = (U-\overline{A})x = Ux - \overline{A}x$, where $U$ is the all-ones matrix
whose linear operator can be computed trivially using $O(n)$ semigroup
operations, and $\overline{A}$ is the complement of~$A$ and therefore has only
$z = O(n)$ ones.

\subsubsection{Commutative Case}
Our first main result shows that in the commutative case, complements
of sparse matrices can be processed as
efficiently as sparse matrices. Specifically, we prove
that if the semigroup is commutative, $Ax$ can be computed in $O(z)$ semigroup
operations; or, more formally, there exists
a~circuit of size $O(z)$ that uses $x=(x_1, \dotsc, x_n)$ as
an~input and computes $Ax$ by only applying the semigroup
operation~$\circ$ (we provide the formal definition of the
computational model in Section~\ref{subsec:circuits}). Moreover,
the constructed circuits are \emph{uniform} in the sense that they
can be generated by an~efficient algorithm. Hence, our circuits
correspond to an~elementary algorithm that uses no tricks like examining the
values $x_j$, i.e., the semigroup operation $\circ$ is applied in a~(carefully
chosen) order that is independent of the specific input~$x$.

\begin{restatable}{theorem}{upperthm}
\label{thm:upperbound}
Let $(S, \circ)$~be a~commutative semigroup,
and $A \in \{0,1\}^{n \times n}$ be a~matrix
with~$z=\Omega(n)$ zeroes.
There exists a~circuit of size $O(z)$ that uses
a~vector $x = (x_1,\ldots, x_n)$ of formal variables as an input,
uses only the semigroup operation~$\circ$ at internal gates,
and outputs $Ax$. Moreover, there exists a~randomized
algorithm that takes the positions of $z$~zeroes of~$A$
as an input and outputs such a~circuit in time $O(z)$
with probability at least $1-\frac{O(\log^5n)}{n}$. There also
exists a~deterministic algorithm with running time $O(z+n\log^4n)$.
\end{restatable}

We state the result for square matrices to simplify the presentation. Theorem~\ref{thm:upperbound} generalizes easily
to show that $Ax$ for a~matrix $A \in \{0,1\}^{m \times n}$ with $z=\Omega(n)$ zeroes can be computed using $O(m+z)$ semigroup operations. Also,
we assume that $z=\Omega(n)$ to be able to state an upper bound $O(z)$ instead of $O(z+n)$. Note that when
$z<n$, the matrix~$A$ is forced to contain all-one rows
that can be computed trivially.

The following corollary generalizes Theorem~\ref{thm:upperbound}
from vectors to matrices.

\begin{restatable}{corollary}{matrixmultcor}
\label{cor:matrixmultiplication}
Let $(S, \circ)$~be a~commutative semigroup.
There exists a~deterministic algorithm that takes
a~matrix $A \in \{0,1\}^{n \times n}$ with
$z=O(n)$~zeroes
and a~matrix $B \in S^{n \times n}$ and computes
the product $AB$ in time $O(n^2)$.
\end{restatable}

\subsubsection{Non-commutative Case}
As our second main result, we show that \emph{commutativity is essential}: for
a~faithful non-commutative semigroup~$S$
(the notion of faithful non-commutative semigroup  is made formal
later in the text), the minimum number of semigroup operations
required to compute $Ax$ for a~matrix
$A \in \{0,1\}^{n \times n}$ with $z=O(n)$ zeroes is
$\Theta(n\alpha(n))$, where $\alpha(n)$ is the inverse Ackermann function.

\begin{restatable}{theorem}{lowerthm}
\label{thm:lowerbound}
Let $(S, \circ)$ be a~faithful non-commutative semigroup, $x = (x_1,\ldots, x_n)$ be
a~vector of formal variables, and $A \in \{0,1\}^{n \times n}$
be a~matrix with $O(n)$ zeroes. Then $Ax$ is computable
using $O(n\alpha(n))$ semigroup operations, where $\alpha(n)$
is the inverse Ackermann function. Moreover, there exists
a~matrix~$A \in \{0,1\}^{n \times n}$ with exactly two zeroes
in every row such that the minimum number of semigroup
operations
required to compute~$Ax$ is $\Omega(n\alpha(n))$.
\end{restatable}

\subsection{Motivation}
The complexity of linear operators is interesting for many reasons.
\begin{description}
\item[Range queries.] In the \emph{range queries} problem,
one is given a~vector~$x=(x_1, \dotsc, x_n)$ over a~semigroup $(S, \circ)$ and
multiple queries of the form~$(l,r)$, and is required to
output the result $x_l \circ x_{l+1} \circ \dotsb \circ x_r$
for each query. It~is a~classical problem in data structures and
algorithms with applications in many fields, such as bioinformatics and
string algorithms, computational geometry, image analysis, real-time
systems, and others. We review some of the less straightforward applications in Section~\ref{subseq:rmqapp},
as well as a~rich variety of algorithmic techniques for the problem in
Section~\ref{subsec:approaches}.

The linear operator problem is a~natural generalization of the range queries
problem: each row of the matrix~$A$ defines a~subset of the elements of~$x$
that need to be summed up and this subset is not required to be
a~contiguous range. The algorithms (Theorem~\ref{thm:upperbound} and
Corollary~\ref{cor:matrixmultiplication}) and hardness results
(Theorem~\ref{thm:lowerbound}) for the linear operator problem presented in this
paper are indeed inspired by some of the known results for the range queries
problem.

\item[Graph algorithms.] Various graph path/reachability
problems can be reduced naturally to matrix multiplication.
Two classic examples are: (i) the all-pairs shortest path problem (APSP) is
reducible to min-plus matrix multiplication, and (ii) the number of triangles
in an undirected graph can be found by computing the third power of its
adjacency matrix.
It is natural to ask what happens if
a~graph has $O(n)$ edges or $O(n)$ anti-edges
(as usual, by~$n$ we denote the number of nodes).
In many cases, an efficient algorithm
for sparse graphs ($O(n)$ edges) is straightforward
whereas an algorithm with the same efficiency
for complements of sparse graphs ($O(n)$ anti-edges) is not. For
example, it is easy to solve APSP and triangle counting on sparse graphs in
time $O(n^2)$, but achieving the same time complexity for complements of sparse
graphs is more complicated.
Theorem~\ref{thm:upperbound} and Corollary~\ref{cor:matrixmultiplication} can be
used to solve various problems on complements of sparse graphs in time $O(n^2)$.

\item[Matrix multiplication over semirings.] Fast matrix
multiplication methods rely essentially on the ring structure of the underlying
set of elements. The first such algorithm was given by~Strassen,
the current record upper bound is $O(n^{2.373})$~\cite{DBLP:conf/stoc/Williams12, DBLP:conf/issac/Gall14a}.
The removal of the inverse operation often drastically increases the complexity
of algorithmic problems over algebraic structures, and even the complexity of
standard computational tasks are not well understood over tropical and
Boolean semirings (see, e.g.~\cite{Williams14,GrigorievP15}).
For various important semirings,
we still do not know an $n^{3-\varepsilon}$ (for a~constant~$\varepsilon>0$) upper
bound for matrix multiplication, e.g., the strongest known upper bound for
min-plus matrix multiplication is $n^3/\exp(\sqrt{\log n})$~\cite{Williams14}.

The interest in computations over such algebraic structures has
recently grew substantially throughout the
Computer Science community with the cases of Boolean and
tropical semirings being of the main interest (see, for
example,~\cite{Jukna16,Williams14,butkovic10systems}).
From this perspective, the computation complexity over sparse and complements of
sparse matrices is one of the most basic questions.
Theorem~\ref{thm:upperbound} and Corollary~\ref{cor:matrixmultiplication}
therefore characterise natural special
cases when efficient computations are possible.

\item[Functional programming.]
Algebraic data structures for graphs developed in the functional programming
community~\cite{mokhov2017algebraic} can be used for representing and processing
densely-connected graphs in linear (in the number of vertices) time and memory.
As we discuss in Section~\ref{sec-dense-graph}, Theorem~\ref{thm:upperbound}
yields an algorithm for deriving a~linear-size algebraic graph representation
for complements of sparse graphs.

\item[Circuit complexity.] Computing linear operators over
a~Boolean semiring~$(\{0,1\}, \lor)$ is a~well-studied problem
in circuit complexity. The corresponding computational model is known
as~\emph{rectifier networks}. An overview of known lower and upper bounds for
such circuits is given by Jukna~\cite[Section~13.6]{DBLP:books/daglib/0028687}.
Theorem~\ref{thm:upperbound} states that very dense linear operators have
linear rectifier network complexity.
\end{description}

\section{Background}
\subsection{Semigroups and Semirings}
A~\emph{semigroup} $(S, \circ)$ is an algebraic structure, where
the operation
$\circ$~is~\emph{closed}, i.e., $\circ : S\times S \rightarrow S$,
and
\emph{associative}, i.e.,
$x \circ (y \circ z) = (x \circ y) \circ z$ for all~$x$, $y$,~and~$z$
in~$S$.
\emph{Commutative} (or \emph{abelian}) semigroups introduce
one extra requirement: $x \circ y = y \circ x$ for all $x$ and $y$
in~$S$.

A~commutative semigroup $(S, \circ)$ can often be extended to
a~\emph{semiring} $(S, \circ, \bullet)$ by introducing
another associative (but not necessarily
commutative)
operation $\bullet$ that \emph{distributes} over~$\circ$, that is
\[
x \bullet (y \circ z) = (x \bullet y) \circ (x \bullet z)\\
\]
\[
(x \circ y) \bullet z = (x \bullet z) \circ (y \bullet z)
\]
hold for all~$x$, $y$,~and~$z$ in~$S$.
Since $\circ$~and~$\bullet$ behave
similarly to numeric addition and multiplication, it is common to
give~$\bullet$ a~higher precedence to avoid
unnecessary parentheses, and even omit~$\bullet$~from
formulas altogether, replacing it by juxtaposition.
This gives a terser and
more convenient notation, e.g., the left distributivity law becomes:
$x (y \circ z) = x y \circ x z$. We will use this notation,
insofar as this does not lead to ambiguity. See Subsection~
\ref{subsec:algstr} for an overview of commonly used
semigroups and semirings.

\subsection{Range Queries Problem and Linear Operator Problem}
In the {\em range queries problem}, one is given
a~sequence $x_1, x_2, \dotsc, x_n$ of
elements of a~fixed semigroup $(S, \circ)$.
Then, a~\emph{range query} is
specified by a~pair of indices $(l,r)$, such that $1 \le l \le r \le n$.
The answer to such a~query is the result of applying the semigroup
operation to the
corresponding range, i.e., $x_l \circ x_{l+1} \circ \dotsb \circ x_r$.
The range queries problem is then to simply answer all given range
queries.
There are two
regimes: online and offline. In the {\em online regime}, one is given
a~sequence of {\em values}
$x_1=v_1, x_2=v_2, \dotsc, x_n=v_n$ and is asked to preprocess
it so that to
answer efficiently any subsequent query.
By ``efficiently'' one usually
means in time independent of the length of the range
(i.e., $r-l+1$, the time
of a~naive algorithm), say, in time $O(\log n)$ or $O(1)$.
In this paper, we
focus on the {\em offline} version, where one is given a~sequence
together with
all the queries, and are interested in the minimum number of
semigroup
operations needed to answer all the queries. Moreover, we study
a~more general
problem: we assume that $x_1, \dotsc, x_n$ are formal variables
rather than
actual semigroup values. That is, we study the {\em circuit size} of
the corresponding
computational problem.

The {\em linear operator} problem generalizes the range
queries problem: now, instead of contiguous ranges one wants
to compute sums over arbitrary subsets. These subsets are
given as rows of a~0/1-matrix~$A$.

\subsection{Circuits}\label{subsec:circuits}
We assume that the input consists of $n$~formal variables
$\{x_1, \dotsc, x_n\}$. We are interested in the minimum number of semigroup
operations needed to compute all given words $\{w_1, \dotsc, w_m\}$ (e.g., for
the range queries problem, each word has a~form $x_l\circ x_{l+1}\circ \dotsb \circ x_r$). We use
the following natural {\em circuit} model. A~circuit computing all these queries
is a~directed acyclic graph. There are exactly $n$~nodes of zero in-degree. They
are labelled with $\{1, \dotsc, n\}$ and are called {\em input gates}. All
other nodes have positive in-degree and are called {\em gates}. Finally, some
$m$~gates have out-degree~0 and are labelled with $\{1, \dotsc, m\}$; they are called {\em output gates}. The
{\em size} of a~circuit is its number of edges (also called {\em wires}). Each
gate of a~circuit computes a~word defined in a~natural way: input gates compute
just $\{x_1, \dotsc, x_n\}$; any other gate of in-degree~$r$ computes a~word
$f_1 \circ f_2 \circ \dotsb \circ f_r$ where $\{f_1, \dotsc, f_r\}$ are words
computed at its predecessors (therefore, we assume that there is an underlying
order on the incoming wires for each gate). We say that the circuit computes the
words $\{w_1, \dotsc, w_m\}$ if the words computed at the output gates are
equivalent to $\{w_1, \dotsc, w_m\}$ over the considered semigroup.

For example, the following circuit computes range queries
$(l_1,r_1)=(1,4)$,
$(l_2,r_2)=(2,5)$, and
$(l_3,r_3)=(4,5)$
over inputs $\{x_1, \dotsc, x_5\}$ or, equivalently, the
linear operator $Ax$ where the matrix $A$~is given below.

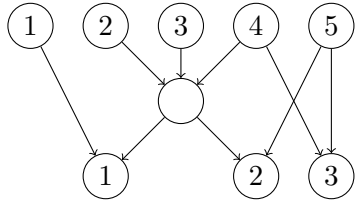
\begin{center}
\begin{tikzpicture}
\foreach \x/\y/\n/\t in {0/4/x1/1, 1/4/x2/2, 2/4/x3/3, 3/4/x4/4, 4/4/x5/5, 2/3/a/~, 1/2/b/1, 3/2/c/2, 4/2/d/3}
  \node[inner sep=0mm,circle,draw,minimum size=6mm] (\n) at (\x,\y) {$\t$};
\foreach \s/\t in {x2/a, x3/a, x4/a, x1/b, a/b, x5/c, a/c, x4/d, x5/d}
  \draw[->] (\s) -- (\t);

\node at (8,3) {$A=\begin{pmatrix}1&1&1&1&0\\0&1&1&1&1\\0&0&0&1&1\end{pmatrix}$};
\end{tikzpicture}
\end{center}

For a~0/1-matrix~$A$, by $C(A)$ we denote the minimum size of
a~circuit computing the linear operator $Ax$.

A~{\em binary circuit} is a~circuit having no gates of fan-in more than two. It
is not difficult to see that any circuit can be converted into a~binary circuit
of size at most twice the size of the original circuit. For this, one just
replaces every gate of fan-in~$k$, for $k>2$, by a~binary tree with $2k-2$ wires
(such a~tree contains $k$~leaves hence $k-1$ inner nodes and $2k-2$ edges).
In~the binary circuit the number of gates does not exceed its size
(i.e., the number of wires). And the number of gates in a~binary
circuit is
exactly the minimum number of semigroup operations needed to
compute the
corresponding function.

We call a~circuit~$C$ computing $A$ \emph{regular} if for every pair $(i,j)$ such that $A_{ij}=1$, there exists exactly one path from the input~$j$ to the output~$i$. A~convenient property of regular circuits is the following observation.

\begin{observation}\label{obs:transpose}
Let $C$~be a~regular circuit computing a~0/1-matrix~$A$ over a~commutative semigroup. Then, by reversing all the wires in~$C$ one gets a~circuit computing~$A^T$.
\end{observation}
Instead of giving a~formal proof, we provide an example of a~reversed circuit from the example given above. It is because of this observation that we require circuit outputs to be gates of out-degree zero (so that when reversing all the wires the inputs and the outputs exchange places).

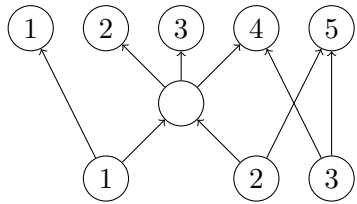
\begin{center}
\begin{tikzpicture}
\foreach \x/\y/\n/\t in {0/4/x1/1, 1/4/x2/2, 2/4/x3/3, 3/4/x4/4, 4/4/x5/5, 2/3/a/~, 1/2/b/1, 3/2/c/2, 4/2/d/3}
  \node[inner sep=0mm,circle,draw,minimum size=6mm] (\n) at (\x,\y) {$\t$};
\foreach \s/\t in {x2/a, x3/a, x4/a, x1/b, a/b, x5/c, a/c, x4/d, x5/d}
  \draw[<-] (\s) -- (\t);

\node at (8,3) {$A^T=\begin{pmatrix}1&0&0\\1&1&0\\1&1&0\\1&1&1\\0&1&1\end{pmatrix}$};
\end{tikzpicture}
\end{center}

\section{Commutative Case}\label{sec-commutative}
This section is devoted to the proofs of Theorem~\ref{thm:upperbound} and Corollary~\ref{cor:matrixmultiplication}
which we remind below.

\upperthm*

\matrixmultcor*

We start by proving two simpler statements to show
how commutativity is important.

\begin{lemma}\label{lemma:easy}
Let $S$~be a semigroup (not necessarily commutative) and let
$A \in \{0,1\}^{n \times n}$ contain at most
one zero in every row. Then
$C(A) = O(n)$.
\end{lemma}
\begin{proof}
To compute the linear operator $Ax$, we first
precompute all prefixes and suffixes of $x=(x_1, \dotsc, x_n)$.
Concretely, let $p_i=x_1 \circ x_2 \circ \dotsb \circ x_i$. All $p_i$'s can be computed
using $(n-1)$ binary gates as follows:
\[
p_1=x_1, p_2=p_1 \circ x_2, p_3=p_2 \circ x_3, \dotsc, p_i=p_{i-1} \circ x_i, \dotsc, p_n=p_{n-1}\circ x_n.
\]
Similarly, we compute all suffixes
$s_j=x_j \circ x_{j+1} \dotsb \circ x_n$ using
$(n-1)$ binary gates. From these prefixes and suffixes
all outputs can be
computed as follows: if a~row of~$A$ contains no zeroes,
the corresponding
output is~$p_n$; otherwise if a~row contains a~zero at position~$i$, the
output is $p_{i-1} \circ s_{i+1}$ (for $i=1$ and $i=n$, we omit the redundant
term).
\end{proof}

In the rest of the section, we assume that the
underlying semigroup is
commutative. Allowing at most two zeroes per row already leads to a~non-trivial
problem. We give only a~sketch of the solution below, since we will further
prove a~more general result. It is interesting to compare the following lemma
with Theorem~\ref{thm:lowerbound} that states that in the
non-commutative setting matrices with two zeroes per row are already hard.

\begin{lemma} \label{lem:at_most_2}
Let $A \in \{0,1\}^{n \times n}$ contain at most two zeroes in every row. Then
$C(A) = O(n)$.
\end{lemma}
\begin{proof}[Proof sketch]
Consider the following undirected graph: the set of nodes is $\{1,2,\dotsc,n\}$;
two nodes $i$ and $j$ are joined by an edge if there is a~row having zeroes in
columns~$i$ and~$j$. In the worst case (all rows are different and contain
exactly two zeroes), the graph has exactly $n$~edges and hence it contains a cut
$(L,R)$ of size at least $n/2$. This cut splits the columns of the matrix into
two parts ($L$ and $R$). Now let us also split the rows into two parts: the top
part $T$~contains all columns that have exactly one zero in each $L$ and $R$;
the bottom part $B$ contains all the remaining rows. What is nice about the top
part of the matrix ($T \times (L \cup R)$) is that it can be computed by $O(n)$
gates (using Lemma~\ref{lemma:easy}). For the bottom part, let us cut all-1
columns out of it and make a recursive call (note that this requires the
commutativity). The corresponding recurrence relation is $T(n) \le cn + T(n/2)$
for a fixed constant $c$, implying $T(n)=O(n)$, and hence $C(A) = O(n)$.
\end{proof}

We now state a~few auxiliary lemmas that will be
used as building blocks in the proof of Theorem~\ref{thm:upperbound}.


\begin{lemma}\label{lemma:decompose}
There exists a~binary regular circuit of size $O(n\log n)$ such that
any range can be computed in a~single additional binary gate
using two gates of the circuit. It can be generated in time
$O(n\log n)$.
\end{lemma}

\begin{lemma}\label{lemma:blocks}
There exists a~binary regular circuit of size $O(n)$ such
that any range
of length at least $\log n$ can be computed in two binary
additional gates from the gates of the circuit.
It can be generated by an algorithm in time $O(n)$.
\end{lemma}

\begin{lemma}\label{lemma:permute}
Let $m \le n$ and $A \in \{0,1\}^{m \times n}$
be a~matrix with $z=\Omega(n)$ zeroes and at most $\log n$ zeroes in every row. There exists a~circuit of size $O(z)$ computing $Ax$. Moreover, there exists a~randomized $O(z)$ time algorithm that takes as input the positions of $z$~zeros
and outputs a~circuit computing $Ax$ with probability at least $1-\frac{O(\log^5n)}{n}$. There also exists a~deterministic
algorithm with running time $O(n\log^4n)$.
\end{lemma}

\begin{proof}[Proof of Theorem~\ref{thm:upperbound}]
Denote the set of rows and the set of columns of~$A$ by~$R$
and~$C$, respectively. Let $R_0 \subseteq R$ be all the rows
having at least $\log n$ zeroes and $R_1=R \setminus R_0$.
We will compute all the ranges of~$A$. From these ranges,
it takes $O(z)$ additional binary gates to compute all the outputs.

We compute the matrices $R_0 \times C$ and $R_1 \times C$
separately. The main idea is that $R_0 \times C$ is easy to compute
because it has a~small number of rows (at most $z/\log n$), while $R_1 \times C$
is easy to compute because it has a~small number of zeroes in every row (at most $\log n$).

The matrix $R_1 \times C$ can be computed using
Lemma~\ref{lemma:permute}. To compute
$R_0 \times C$,
it suffices to compute $C \times R_0$ by a~regular circuit,
thanks to the Observation~\ref{obs:transpose}.
Let $|R_0|=t$. Clearly, $t \le z/\log n$.
Using Lemma~\ref{lemma:decompose}, one can compute all
ranges of $C \times R_0$ by a~circuit of size
\[O(t\log t+z)=O\left(\frac{z}{\log n} \cdot \log z+z\right)=O(z+n)=O(z)\, ,\]
since $z =O(n^2)$.

The algorithm for generating the circuit is just a~combination
of the algorithms from Lemmas~\ref{lemma:decompose} and~\ref{lemma:permute}.
\end{proof}

\begin{proof}[Proof of Lemma~\ref{lemma:decompose}]
We adopt the divide-and-conquer construction by~Alon and Schieber~\cite{Alon87optimalpreprocessing}.
Split the input range $(1,n)$ into two half-ranges of
length~$n/2$:
$(1,n/2)$ and $(n/2+1,n)$.
Compute all suffixes of the left half and all prefixes of
the right half.
Using these precomputed suffixes and
prefixes one can answer any query $(l,r)$ such that $l \le n/2
\le r$ in a~single additional gate. It remains to be able to answer
queries that lie entirely in one of the halves. We do this by
constructing recursively circuits for both halves. The resulting
recurrence relation $T(n) \le 2T(n/2)+O(n)$ implies that the
resulting circuit has size at most $O(n\log n)$.
\end{proof}

\begin{proof}[Proof of Lemma~\ref{lemma:blocks}]
We use the block decomposition technique for
constructing the required circuit.
Partition the input range $(1,n)$ into $n/\log n$ ranges
of length $\log n$ and call them blocks. Compute the range
corresponding to each block (in total size $O(n)$).
Build a~circuit from Lemma~\ref{lemma:decompose} on
top of these blocks. The size of this circuit is $O(n)$ since the
number of blocks is $n/\log n$.
Compute all prefixes and all suffixes of every block. Since
the blocks partition the input range $(1,n)$, this also can be done
with an $O(n)$ size circuit.

Consider any range of length at least $\log n$. Note that it
cannot lie entirely inside the block. Hence, any such range can be
decomposed into three subranges: a~suffix of a~block, a~range
of blocks, and a~prefix of a~block
(where any of the three components may be empty). For example, for $n=16$,
a~range $(3,13)$ is decomposed into a~suffix $(3,4)$ of the
first block,
a~range $(2,3)$ of blocks $(B_1, B_2, B_3, B_4)$, and a~prefix $(13,13)$ of
the last block:
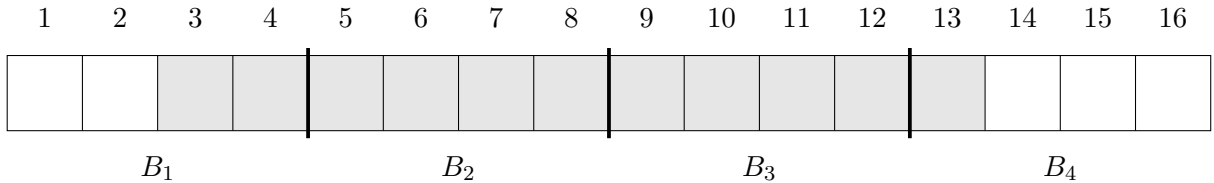
\begin{center}
\begin{tikzpicture}
\foreach \x in {1,...,16}
  \node at (\x,2) {\x};
\draw[draw=white,fill=gray!20!white] (2.5,0.5) rectangle (13.5,1.5);
\foreach \x in {1,...,15}
  \draw (\x+0.5,0.5) -- (\x+0.5,1.5);
\draw (0.5,0.5) rectangle (16.5,1.5);
\foreach \x in {4,8,12}
  \draw[line width=.5mm] (\x+0.5,0.4) -- (\x+0.5,1.6);
\foreach \x/\i in {2/1, 6/2, 10/3, 14/4}
  \node at (\x+0.5,0) {$B_{\i}$};
\end{tikzpicture}
\end{center}
It remains to note that all these three components are already precomputed.
\end{proof}

\begin{proof}[Proof of Lemma~\ref{lemma:permute}.]
The $z$~zeroes of~$A$ breaks its rows into ranges.
Let us call a~range {\em short} is its length is at most $\log n$.
We will show that it is possible to permute the columns of~$A$
so that the total length of all short ranges is at most $O(\frac{n}{\log n})$. Then, all such short ranges can be computed by a~circuit of size $O(\frac{\log n}{n} \cdot n)=O(n)=O(z)$.
All the remaining ranges can be computed by a~circuit of size $O(n)$ using Lemma~\ref{lemma:blocks}.

\begin{description}
\item[Randomized algorithm.]

Permute the columns randomly. A~uniform random permutation
of $n$~objects can be generated in time~$O(n)$~\cite[Algorithm~P (Shuffling)]{DBLP:books/lib/Knuth98}.
Let us compute the expectation of
the total length of short ranges.
Let us focus on a~single row and a~particular cell in it. Denote the number of zeroes in the row by~$t$. What is the probability that the cell belongs to a~short segment? There are two cases to consider.
\begin{enumerate}
\item The cell lies close to the border, i.e., it belongs to
the first $\log n$ cells or to the last~$\log n$ cells
(the number of such cells is $2\log n$). Then,
this cell belongs to a~short range iff there is at least one zero
in $\log n$ cells close to it (on the side opposite to the border).
Hence, one zero must belong to a~set of $\log n$ cells while the remaining $t-1$ zeroes may be anywhere.
The probability is then at most
\[\log n \cdot \frac{\binom{n}{t-1}}{\binom{n}{t}}=\log n \cdot \frac{t}{n-t+1}=O\left(\frac{\log^2n}{n}\right) \, .\]
\item It is not close to the border (the number of such cells is $n-2\log n$). Then, there must be a~zero on both sides of the
cell. The probability is then at most
\[\log^2 n \cdot \frac{\binom{n}{t-2}}{\binom{n}{t}}=\log^2n \cdot \frac{t(t-1)}{(n-t+1)(n-t+2)}=O\left(\frac{\log^4 n}{n^2}\right) \, .\]
\end{enumerate}
Hence, the expected total length of short ranges in one row is
\[O\left( 2\log n \cdot \frac{\log^2 n}{n} + (n-2\log n) \cdot \frac{\log^4 n}{n^2}\right)=O\left(\frac{\log^4 n}{n}\right) \, .\]
Thus, the expected length of short ranges in the whole
matrix~$A$ is $O(\log^4n)$. By Markov inequality, the probability that
the length of all short ranges is larger than $\frac{n}{\log n}$ is
at most $O(\frac{\log^5 n}{n})$.

\item[Deterministic algorithm.]
It will prove convenient to assume that $A$ is a~$t \times t$ matrix with exactly~$t$ zeros with at most $\log t$ zeroes in every row. For this, we let $t=\max\{n, z\}$ and add a~number of
all-ones rows and columns if needed. This enlargement of the matrix
does not make the computation simpler: additional rows mean additional outputs that can be ignored and additional columns correspond to redundant variables that can be removed (substituted by~0) once the circuit is constructed. Below, we show how to deterministically construct a~circuit of size $O(t)$ for~$A$.
For this, we present a~greedy algorithm for permuting the columns
of~$A$ in such a~way that the total length of all short segments
is $O(\log^4n)$. This will follow from the fact that all short
ranges in the resulting matrix~$A$ will lie within the last $O(\log^2 t)$ columns.

We construct the required permutation of columns step by step by a~greedy algorithm. After step~$r$ we will have a~sequence of the first~$r$ columns chosen and we will maintain the following properties:
\begin{itemize}
\item For each $i \leq r$, the first $i$~columns contain at least $i$ zeros.
\item There are no short ranges within the first~$r$ rows (apart from those, that can be extended by adding columns on the right).
\end{itemize}
The process will work for at least $t - \log^2 t$ steps, so short ranges are only possible within the last $\log^2 t + \log t = O(\log^2 t)$ columns.

On the first step, we pick any column that has a zero in it. Suppose
we have reached step~$r$. We explain how to add a~column on
step $r+1$. Consider the last $\log t$ columns in the currently
constructed sequence. Consider the set $R$~of rows that have
zeros in them. These are exactly the rows that constrain our
choice for the next column. There are two cases.
\begin{enumerate}
\item There are at most $\log t$ rows in~$R$. Then, for each row in~$R$, there are at most $\log t$ columns that have zeros in this row. In total, there are at most $\log^2 t$ columns that have zeros in some row of~$R$. Denote the set of this columns by~$F$. If there is an unpicked column outside of~$F$ that has at least one zero in it, we add this column to our sequence. Clearly, both properties are satisfied and the step is over. Otherwise, all other columns contain only ones, so we add all of them to our sequence, place the columns from~$F$ to the end of the sequence, and the whole permutation is constructed.
\item There are more that $\log t$ rows in~$R$. This means that the last $\log t$ columns of the current sequence contain more than $\log t$ zeros. By the first property, the first $r - \log t$ columns contain at least $r - \log t$ zeros. So overall, in the current sequence of~$r$ columns there are more than $r$~zeros. Thus, in the remaining $t-r$ columns there are less then $t-r$ zeros and there is a~column without zeros. We add this column to the sequence.
\end{enumerate}
\end{description}

To implement this algorithm in time $O(t\log^{4}t)$, we store, for each column~$j$ of~$A$, a~sorted array of rows~$i$ such that $A_{ij}=0$. Since the total number of zeros~$z$ is at most~$t\log t$, these arrays can be computed in time $O(t\log^2t)$: if $c_1, \dotsc, c_t$ are the numbers of zeros in the columns, then sorting the corresponding arrays takes time
\[\sum_{i=1}^{t}c_i \log c_i \le \log(t \log t) \cdot \sum_{i=1}^{t}c_i \le \log(t \log t) \cdot t\log t \, .\]

At every iteration, we need to update the set~$R$. For this, we need to remove some rows from it (from the column that no longer belongs to the stripe of columns of width $\log t$) and to add the rows of the newly added column. Since the size of~$|R|$ is always at most~$t$ and the total number of zeros is $z \le t\log t$, the total running time for all such updates is
$O(t\log^2t)$ (if one uses, e.g., a~balanced binary search tree for
representing~$R$).

If $|R| > \log t$, one just takes an all-one column (all such columns
can be stored in a~list). If $|R| \le \log t$, we need to find a~column outside of the set~$F$. For this, we just scan the list of the yet unpicked columns. For each column, we first check whether it belongs to the set~$F$. This can be checked in time $O(\log^2t)$: for every row in~$|R|$, one checks whether this row belongs to the sorted array of the considered column using binary search in time
$O(\log t)$. Since $|F| \le \log^2t$, we will find a~column outside of~$F$ in time~$O(\log^4 t)$.
\end{proof}

\begin{proof}[Proof of Corollary~\ref{cor:matrixmultiplication}]
One deterministically generates a~circuit for~$A$ of size $O(n)$ in time $O(n\log^4n)=O(n^2)$ by Theorem~\ref{thm:upperbound}.
This circuit can be used to multiply~$A$ by any column of~$B$
in time~$O(n)$. For this, one constructs a~topological ordering of the gates of the circuits and computes the values of all gates in
this order. Hence, $AB$ can be computed in time~$O(n^2)$.
\end{proof}

\section{Non-commutative Case}\label{sec-non-commutative}

In the previous section, we have shown that for commutative semigroups dense linear operators can be computed by linear size circuits. A~closer look at the circuit constructions reveals that we use commutativity crucially: it is important that we may reorder the columns of the matrix. In this section, we show that this trick is unavoidable: for non-commutative semigroups, it is not possible to construct linear size circuits for dense linear operators. Namely, we prove Theorem~\ref{thm:lowerbound}.

\lowerthm*

\subsection{Faithful semigroups}

We consider computations over general semigroups that are not necessarily commutative. In particular, we will establish lower bounds for a large class of semigroups and our lower bound does not hold for commutative semigroups. This requires a formal definition that captures semigroups with rich enough structure and in particular requires that a semigroup is substantially non-commutative.

Previously lower bounds in the circuit model for a large class of semigroups were known for the Range Queries problem~\cite{DBLP:conf/stoc/Yao82,DBLP:journals/ijcga/ChazelleR91}. These result are proven for a large class of commutative semigroups that are called \emph{faithful} (we provide a formal definition below). Since we are dealing with non-commutative case we need to generalize the notion of faithfulness to non-commutative semigroups.

To provide formal definition of faithfulness it is convenient to introduce the following notation.
Suppose $(S, \circ)$ is a
semigroup. Let $X_{S,n}$ be a semigroup with generators $\{x_1,\ldots, x_n\}$
and with the equivalence relation consisting of identities in variables
$\{x_1,\ldots, x_n\}$ over~$(S,\circ)$. That is, for two words $W$ and $W'$ in
the alphabet $\{x_1,\ldots,x_n\}$ we have $W\sim W'$ in $X_{S,n}$ iff no matter which
elements of the semigroup~$S$ we substitute for $\{x_1,\ldots, x_n\}$ we obtain
a~correct equation over~$S$. In particular, note that if $S$~is commutative
(respectively, idempotent), then $X_{S,n}$ is also commutative (respectively,
idempotent).
The semigroup $X_{S,n}$ is studied in algebra under the name of relatively free semigroup of rank $n$ of a variety generated by semigroup $S$~\cite{neumann2012varieties}.
We will often omit the subscript $n$ and write simply $X_S$ since the number of generators will be clear from the context.

Below we will use the following notation. Let $W$ be a word in the alphabet
$\{x_1,\ldots, x_n\}$. Denote by $\var(W)$ the set of letters that are present
in $W$.


We are now ready to introduce the definition of a commutative faithful semigroup.

\begin{definition}[\cite{DBLP:conf/stoc/Yao82,DBLP:journals/ijcga/ChazelleR91}]
A commutative semigroup $(S, \circ)$ is \emph{faithful commutative} if for any
equivalence $W\sim W'$ in $X_S$ we have $\var(W)=\var(W')$.
\end{definition}

Note that this definition does not pose any restrictions on the cardinality of
each letter in $W$ and $W'$. This allows to capture in this definition
important cases of idempotent semigroups. For example, semigroups
$(\{0,1\}, \vee)$ and $(\mathbb{Z},\min)$ are commutative faithful.

We need to study the non-commutative case, and moreover, our results
establish the difference between commutative and non-commutative cases. Thus,
we need to extend the notion of faithfulness to non-commutative semigroups to
capture their non-commutativity in the whole power. At the same time we would
like to keep the case of idempotency. We introduce the notion of faithfulness
for the non-commutative case inspired by the properties of free idempotent
semigroups~\cite{GreenR52}. To introduce this notion we need several
definitions.

The \emph{initial mark} of $W$ is the letter that is present in $W$ such that
its first appearance is farthest to the right. Let $U$ be the prefix of $W$
consisting of letters preceding the initial mark. That is, $U$ is the maximal
prefix of $W$ with a smaller number of generators. We call $U$ the
\emph{initial} of $W$. Analogously we define the \emph{terminal mark} of $W$ and
the \emph{terminal} of $W$.

\begin{definition}\label{def:strong_non_commutativity}
We say that a semigroup $X$ with generators $\{x_1,\ldots, x_n\}$ is
\emph{strongly non-commutative} if for any words $W$ and $W'$ in the
alphabet $\{x_1,\ldots, x_n\}$ the equivalence $W\sim W'$ holds in $X$ only if
the initial marks of $W$ and $W'$ are the same, terminal marks are the same,
the equivalence $U \sim U'$ holds in $X$, where $U$ and $U'$ are the initials of
$W$ and $W'$, respectively, and the equivalence $V \sim V'$ holds in $X$, where
$V$ and $V'$ are the terminals of $W$ and $W'$, respectively.
\end{definition}

In other words, this definition states that the first and the last occurrences
of generators in the equivalence separates the parts of the equivalence that
cannot be affected by the rest of the generators and must therefore be equivalent themselves. We also note that this definition exactly captures the
idempotent case: for a free idempotent semigroup the condition in this
definition is ``if and only if''\cite{GreenR52}.

\begin{definition} \label{def:faithful}
A semigroup $(S, \circ)$ is \emph{faithful non-commutative} if $X_S$ is strongly non-commutative.
\end{definition}

We note that this notion of faithfulness is relatively general and is true for
semigroups $(S,\circ)$ with considerable degree of non-commutativity in their
structure. It clearly captures free semigroups with at least two generators. It is also easy to see that the
requirements in Definition~\ref{def:faithful} are satisfied for the free
idempotent semigroup with $n$ generators (if $S$ is idempotent, then $X_{S,n}$ is also
clearly idempotent and no other relations are holding in $X_{S,n}$ since we can
substitute generators of $S$ for $x_1, \ldots, x_n$).

Next we observe some properties of strongly non-commutative semigroups that we
need in our constructions.

\begin{lemma} \label{lem:prefix_equivalence}
Suppose $X$ is strongly non-commutative. Suppose the equivalence $W \sim W'$
holds in~$X$ and $|\var(W)|=|\var(W')|=k$. Suppose $U$~and~$U'$ are minimal
(maximal) prefixes of $W$ and $W'$ such that $|\var(U)| = |\var(U')| = l\leq k$.
Then the equivalence $U \sim U'$ holds in $X$. The same is true for suffixes.
\end{lemma}

\begin{proof}
The proof is by induction on the decreasing $l$. Consider the maximal prefixes
first. For $l=k$ and maximal prefixes we just have $U=W$ and $U'=W'$. Suppose
the statement is true for some $l$, and denote the corresponding prefixes by $U$
and $U'$, respectively. Then note that the maximal prefixes with $l-1$ variables
are initials of $U$ and $U'$. And the statement follows by
Definition~\ref{def:strong_non_commutativity}.

The proof of the statement for minimal prefixes is completely analogous. Note
that on the step of induction the prefixes differ from the previous case by one
letter that are initial marks of the corresponding prefixes. So these additional
letters are also equal by the Definition~\ref{def:strong_non_commutativity}.

The case of suffixes is completely analogous.
\end{proof}

The next lemma is a simple corollary of Lemma~\ref{lem:prefix_equivalence}.
\begin{lemma} \label{lem:variables_order}
Suppose $X$ is strongly non-commutative. Suppose $W \sim W'$ holds in $X$. Let us write down the letters of $W$ in the order in which they appear first time in $W$ when we read it from left to right. Let's do the same for $W'$. Then we obtain exactly the same sequences of letters.

The same is true if we read the words from right to left.
\end{lemma}

\subsection{Proof Strategy}



We now proceed to the proof of Theorem~\ref{thm:lowerbound}.

The upper bound follows easily by a naive algorithm: split all rows of $A$ into ranges, compute all ranges by a circuit of size $O(n\alpha(n))$ using Yao's construction~\cite{DBLP:conf/stoc/Yao82}, then combine ranges into rows of $A$ using $O(n)$ gates.

Thus we will concentrate on lower bounds. We will view the computation of the circuit as a computation in a strongly non-commutative semigroup $X=X_S$.

We will use the following proof strategy.
First we observe that it is enough to prove the lower bound for the case of idempotent strongly non-commutative semigroups $X$. Indeed, if $X$ is not idempotent, we can factorize it by idempotency relations and obtain a strongly non-commutative idempotent semigroup $X_{id}$. A lower bound for the case of $X_{id}$ implies lower bound for the case of $X$. We provide a detailed explanation in Section~\ref{sec:noncommutative_extension}.

Hence, { from this point we can assume that $X$ is idempotent and strongly non-commutative}.
Next for idempotent case we show that our problem is equivalent to the commutative version of the range query problem.

For a semigroup $X$ with generators $\{x_1,\ldots, x_n\}$ denote by $X_{sym}$ its factorization under commutativity relations $x_i x_j \sim x_j x_i$ for all $i,j$. Note that if $X$ is idempotent and strongly non-commutative, then $X_{sym}$ is just the semigroup in which $W \sim W'$ iff $\var(W)=\var(W')$ (this is free idempotent commutative semigroup).

\begin{theorem}\label{thm:equivalence}
For an idempotent strongly non-commutative $X$ and for any $s=\Omega(n)$ we have that (commutative) range queries problem over $X_{sym}$ has size $O(s)$ circuits iff (non-commutative) dense linear operator problem over $X$ has size $O(s)$ circuits.
\end{theorem}

Using this theorem, it is straightforward to finish the proof of
Theorem~\ref{thm:lowerbound}.
Indeed, by Theorem~\ref{thm:equivalence} if non-commutative dense linear operator problem has size $s$ circuit, then the commutative range queries problem also does. However, for the latter problem it is proved by Chazelle and Rosenberg~\cite{DBLP:journals/ijcga/ChazelleR91} that $s=\Omega(n \alpha(n))$. Moreover, in our construction for the proof of Theorem~\ref{thm:equivalence} it is enough to consider dense linear operators with exactly two zeroes in every row. From this the second part of Theroem~\ref{thm:lowerbound} follows.

Note that for the proof of Theorem~\ref{thm:lowerbound} only one direction of Theorem~\ref{thm:equivalence} is needed. However, we think that the equivalence in Theorem~\ref{thm:equivalence} might be of independent interest, so we provide the proof for both directions.

Thus, it remains to prove Theorem~\ref{thm:equivalence}. We do this by showing the following equivalences for any $s = \Omega(n)$.

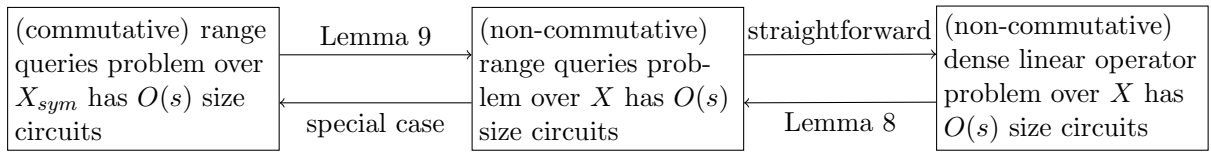
\begin{center}
\begin{tikzpicture}[scale=0.95,transform shape]
\tikzstyle{v}=[rectangle,draw,inner sep=1mm,text width=36mm,above right,minimum height=20mm]

\node[v] (a) at (0,0) {(commutative) range queries problem over $X_{sym}$ has $O(s)$ size circuits};

\node[v] (b) at (6.5,0) {(non-commutative) range queries problem over $X$ has $O(s)$ size circuits};

\node[v] (c) at (13,0) {(non-commutative) dense linear operator problem over $X$ has $O(s)$ size circuits};

\path (a.10) edge[->] node[above] {Lemma~\ref{lem:intervals}} (b.170);
\path (b.190) edge[->] node[below] {special case} (a.-10);
\path (b.10) edge[->] node[above] {straightforward} (c.170);
\path (c.190) edge[->] node[below] {Lemma~\ref{lem:dense_matrices}} (b.-10);
\end{tikzpicture}
\end{center}


Note that two of the reductions on this diagram are trivial. The other two are formulated in the following lemmas.

%
%

\begin{lemma} \label{lem:dense_matrices}
If the (non-commutative) dense linear operator problem over $X$ has size $s$ circuit then the (non-commutative) range queries problem over $X$ has size $O(s)$ circuit.
\end{lemma}

\begin{lemma} \label{lem:intervals}
If the (commutative) version of the range queries problem over $X_{sym}$ has size $s$ circuits then the (non-commutative) version over $X$ also does.
\end{lemma}

The proofs of these lemmas are presented in Sections~\ref{sec:operators_to_queries} and~\ref{sec:non-commutative_to_commutative} respectively.


\subsection{From Idempotent Semigroups to General Semigroups}\label{sec:noncommutative_extension}

In this section we provide a detailed explanation of the reduction in Theorem~\ref{thm:lowerbound} from general semigroups to idempotent semigroups.

Consider an arbitrary strongly non-commutative semigroup $X$. Consider a new semigroup $X_{id}$ over the same set of generators that is a factorization of $X$ by idempotency relations $W^2\sim W$ for all words $W$ in the alphabet $\{x_1,\ldots, x_n\}$.

\begin{lemma} \label{lem:idempotisation}
If $X$ is strongly non-commutative, then $X_{id}$ is also strongly non-commutative.
\end{lemma}

\begin{proof}
Suppose $W$ and $W'$ are words in the alphabet $\{x_1,\ldots, x_n\}$ and $W \sim W'$ in $X_{id}$. This means that there is a sequence $W_0,\ldots, W_k$ of words in the same alphabet such that $W=W_0$, $W'=W_k$ and for each $i$ either $W_i \sim W_{i+1}$ in $X$, or $W_{i+1}$ is obtained from $W_i$ by one application of the idempotency equivalence to some subword of $W_i$. Clearly, it is enough to check that the conditions of Definition~\ref{def:strong_non_commutativity} are satisfied in $X_{id}$ for each consecutive pair $W_i$ and $W_{i+1}$.

If $W_i \sim W_{i+1}$ in $X$, then the conditions of Definition~\ref{def:strong_non_commutativity} follows from the strong non-commutativity of $X$.

Suppose now that $W_{i+1}$ is obtained from $W_{i}$ by substituting some subword $A$ by $A^2$ (the symmetrical case is analyzed in the same way). We will show that initial marks of $W_i$ and $W_{i+1}$ are the same and $U_{i} \sim U_{i+1}$ in $X_{id}$, where $U_{i}$ and $U_{i+1}$ are initials of $W_i$ and $W_{i+1}$ respectively. For the terminals and terminal marks the proof is completely analogous.

Suppose $A$ lies to the left of initial mark in $W_i$ and we substitute $A$ by $A^2$. Then the initial mark is unaltered and in the initial $U_i$ we also substitute $A$ by $A^2$. Thus in this case $U_{i+1}$ is obtained from $U_i$ by idempotency relation.

Suppose $A$ contains initial mark of $W_i$ or lies to the right of it. Then after the substitution of $A$ by $A^2$ the initial mark is still the same and the initial $U_i$ also does not change.
\end{proof}

Now we outline the reduction of the lower bound in Theorem~\ref{thm:lowerbound} from idempotent semigroup to the general case.

Suppose $X$ is strongly non-commutative and suppose that for $X$ all dense operators can be computed by circuits of size at most $s$.

Consider a semigroup $X_{id}$ as introduced above. By Lemma~\ref{lem:idempotisation} $X_{id}$ is also strongly non-commutative. On the other hand, since $X_{id}$ is a factorization of $X$ any circuit computing dense operator over $X$ also computes the same dense operator over $X_{id}$. Thus, by our assumption there are circuits of size at most $s$ for all dense operators over $X_{id}$. Finally, $X_{id}$ is idempotent, so by the special case of our theorem we have $s = \Omega(n \alpha(n))$ and we are done.

\subsection{Reducing Dense Linear Operator to Range Queries} \label{sec:operators_to_queries}
In this subsection, we prove Lemma~\ref{lem:dense_matrices}. Intuitively, the lemma holds as the best way to compute rows of a~dense matrix is to combine input variables in the natural order. This intuition is formalized in Lemma~\ref{lemma:correctorder} below. Given this, it is easy to reduce dense linear operator problem to the range queries problem: we just ``pack'' each range query into a~separate row, i.e., for a~query $(l,r)$ we introduce a~$0/1$-row having two zeroes in positions $l-1$ and $r+1$ (hence, this row consists of three ranges: $(1,l-1)$, $(l,r)$, $(r+1,n)$). Then, if a~circuit computing the corresponding linear operator has a~nice property of always using the natural order of variables (guaranteed by Lemma~\ref{lemma:correctorder}), one may extract the answer to the query $(l,r)$ from it.

It should be mentioned, at the same time, that the semigroup $X$ might be complicated. In particular, the idempotency is tricky. For example, it can be used to simulate commutativity: one can turn $xy$ into $yx$, by first multiplying $xy$ by~$y$ from the left and then multiplying the result by $x$ from the right (obtaining $(y(xy))x=(yx)(yx)=yx$). Using similar ideas, one can place new variables inside of already computed products. To get $xyz$ from $xz$, one multiplies it by $xyz$ first from the left and then from the right: $(xyz)xz(xyz)=xy(zxzx)yz=xy(zx)yz=xyz$.
This is not extremely impressive, since to get $xyz$ we multiply by $xyz$, but the point is that this is possible in principle.

We proceed to the formal proofs. Let's call the word $W$ in the alphabet $\{x_1,\ldots,x_n\}$ \emph{increasing} if it is a product of variables in the increasing order. A~binary circuit is called an~{\em increasing circuit} if each of its gates computes a word equivalent in $X$ to increasing~word.
Note that if a~gate in an~increasing circuit is fed by two gates~$G$ and~$H$, then the increasing words computed by~$G$ and~$H$ are matching in a~sense that some suffix of~$G$ (possibly an empty suffix) is equal to some prefix of~$H$. Otherwise, the result is not equal to a product of variables in the increasing order, due to Lemma~\ref{lem:variables_order}.

Analogously, a~binary circuit is called a~{\em range circuit} if each of its gates computes a~word that is equivalent to a~range.

The proof of Lemma~\ref{lem:dense_matrices} follows from the following two lemmas.

\begin{lemma}\label{lemma:correctorder}
Given a~binary circuit computing~$Ax$, one may transform it into an~increasing circuit of the same size computing the same function.
\end{lemma}

\begin{lemma}\label{lemma:matrixranges}
Given an~increasing circuit computing~$Ax$, one may transform it into a~range circuit of the same size computing all ranges of~$A$.
\end{lemma}

\begin{proof}[Proof of Lemma~\ref{lem:dense_matrices}]
Given $n$~ranges, pack them into a~matrix $A \in \{0,1\}^{n \times n}$ with at most $2n$ zeroes. Take a size-$s$ circuit computing $Ax$ and convert it into a~binary circuit. Then, transform it into an~increasing circuit using Lemma~\ref{lemma:correctorder}. Finally, extract the answers to all the ranges from this circuit using Lemma~\ref{lemma:matrixranges}.
\end{proof}

Note that the second statement of Theorem~\ref{thm:lowerbound} follows since the proof of Lemma~\ref{lem:dense_matrices} deals with matrices with exactly two zeroes in every row.

\begin{proof}[Proof of Lemma~\ref{lemma:matrixranges}]
Take an~increasing circuit ${\cal C}$ computing $Ax$ and process all its gates in some topological ordering. Each gate $G$ of ${\cal C}$ computes a (word that is equivalent to an) increasing word. We split this increasing word into ranges and we put into correspondence to $G$ an ordered sequence $G_1,\ldots, G_k$ of gates of the new circuit. Each of this gates compute one of the ranges of the word computed by $G$ and $G \sim G_1\circ\ldots \circ G_k$.

Consider a gate $G$ of ${\cal C}$ and suppose we have already computed all gates of the new circuit corresponding to previous gates of ${\cal C}$. $G$ is the product $F \circ H$ of previous gates of ${\cal C}$, for which new range gates are already computed. Since ${\cal C}$ is increasing we have that $F$ and $H$ are matching, that is some suffix (maybe empty) of the increasing word computed in $F$ is equal to some prefix (maybe empty) of the increasing word computed in $H$ and there are no other common variables in these increasing words. It is easy to see that ranges for the sequence corresponding to $G$ are just the ranges for the sequences for $F$ and $H$ with possibly two of them united. If needed, we compute the product of gates of the new circuit corresponding to the united ranges and the sequence of new gates for $G$ is ready.

Thus, to process each gate of ${\cal C}$ we need at most one operation in the new circuit and the size of the new circuit is at most the size of ${\cal C}$.

For output gates of ${\cal C}$ we have gates in the new circuit that compute exactly ranges of output gates. Thus, in the new circuit all ranges of $A$ are computed.
\end{proof}

\begin{proof}[Proof of Lemma~\ref{lemma:correctorder}]
Consider a~binary circuit~${\cal C}$ computing~$Ax$ and its gate~$G$ together with a~variable~$x_i$ it depends on.
We say that $x_i$ is \emph{good} in~$G$ if there is
a~path in~${\cal C}$ from $G$ to an output gate, on which the word is never multiplied from the left by words containing variables greater than or equal to $x_i$.
Note that if $x_i$ and $x_{i'}$ are both contained in $G$, $i<i'$, and $x_i$ is good in~$G$, then $x_{i'}$ is good in~$G$, too. That is, the set of all good variables in~$G$ is closed upwards.

Consider the largest good variable in $G$ (if there is one), denote it by $x_k$ ($x_k$ is actually just the largest variable in~$G$, unless of course there are no good variables in $G$). Let us focus on the first occurrence of $x_k$ in~$G$.

\begin{claim}
All first occurrences of other good variables in~$G$ must be to the left of the first occurrence of $x_k$.
\end{claim}

\begin{proof}
Suppose that a~good variable $x_i$ has the first occurrence to the right of (the first occurrence of) $x_k$. Consider an output gate $H$ such that there is a~path from~$G$ to~$H$ and along this path there are no multiplications of $G$ from the left by words containing variables greater than~$x_i$. Then we have $H \sim LGR$, where all variables of~$L$ are smaller then~$x_i$. Then in $H$ the variable $x_i$ appears before $x_k$ when we read from left to right, but at the same time we have that $x_k$ appears before $x_i$ in $LGR$. This contradicts Lemma~\ref{lem:variables_order}.
\end{proof}

Now, for a~gate~$G$, define two words $\mmin_G$ and $\mmax_G$. Both these words are products of variables in the increasing order: $\mmin_G$ is the product of good variables of $G$ in the increasing order, $\mmax_G$ is the product (in the increasing order) of all variables that has first occurrences before (the first occurrence of) $x_k$. Note that $\mmin_G$ is
a~suffix of $\mmax_G$. If there are no good variables in $G$ we just let $\mmin_g=\mmax_g=\lambda$ (the empty word).
For the word~$W$ that has the form of the product of variables in the increasing order, we call $x_j$ a~\emph{gap variable} if it is not contained in $W$
while $W$~contains variables $x_i$ and $x_k$ with $i < j < k$.

Below we show how for a~given circuit~${\cal C}$ to construct an increasing circuit~${\cal C}'$ that for each gate~$G$ of~${\cal C}$ computes some intermediate product $P_G$ between $\mmin_G$ and $\mmax_G$: $\mmin_g$ is a~suffix of $P_G$ and $P_G$ is a~suffix of $\mmax_g$. The size of~${\cal C}'$ is at most the size of~${\cal C}$. For an output gate~$G$, $\mmin_g=\mmax_g=g$ hence the circuit ${\cal C}'$ computes the correct outputs.

To construct ${\cal C}'$, we process the gates of~${\cal C}$ in a~topological ordering. If $G$~is an input gate, everything is straightforward: in this case $\mmax_G=\mmin_G$ is either $\lambda$ or $x_j$. Assume now that $G$ is an internal gate with predecessors~$F$ and~$H$.
Consider the set of good variables in~$G$. If there are none, we let $P_G=\lambda$. If all first occurrences of good variables of $G$ are lying in one of the predecessors ($F$~and~$H$), then they are good in the corresponding input gate. We then set $P_G$ to $P_F$ or $P_H$.

The only remaining case is that some good variables have their first occurrence in~$F$ while some others have their first occurrence in~$H$. Then the largest variable $x_k$ of~$G$ has the first occurrence in $H$ and all variables of~$F$ are smaller than~$x_k$.

\begin{claim} \label{cl: h is good}
There are no gap variables for $\mmax_H$ in~$F$.
\end{claim}

\begin{proof}
Suppose that some variable $x_i$ in~$F$ is a~gap variable for $\mmax_H$. Consider an output $C$ such that there is a path from~$G$ to~$C$ and along this path there are no multiplications of $G$ from the left by words containing variables greater than~$x_k$. Then we have $C \sim LGR$ where all variables of $L$ are smaller then~$x_k$. Consider the prefix $P$ of $C$ preceding the variable~$x_k$ and the prefix~$Q$ of $LG$ preceding the variable $x_k$.
Then by Lemma~\ref{lem:prefix_equivalence} we have $P \sim Q$. But then the variables of~$P$ and~$Q$ appear in the same order if we read the words from right to left. But this is not true (the variable in~$P$ are in the decreasing order and in~$Q$ the variable $x_i$ is not on its place), a~contradiction.
\end{proof}

\begin{claim}\label{cl: f is good}
There are no gap variables for $\mmax_F$ in~$H$.
\end{claim}

\begin{proof}
Suppose that a~variable $x_i$ in~$H$ is a~gap variable for $\mmax_F$. Consider an output~$C$ such that there is a~path from~$G$ to~$C$ and along this path there are no multiplications of~$G$ from the left by words containing variables greater than $x_l$, the largest variable of~$F$. Then we have $C \sim LGR$, where all variables of~$L$ are smaller then $x_l$. Consider the prefix~$P$ of~$C$ preceding $x_l$ and the prefix $Q$ of $LG$ preceding $x_l$.
Then by Lemma~\ref{lem:prefix_equivalence} we have $P \sim Q$. But then the variables of~$P$ and~$Q$ appear in the same order if we read the words from right to left. But this is not true (the variables in~$P$ are in the decreasing order and in~$Q$ the variable $x_i$ is not on its place), a~contradiction.
\end{proof}

We are now ready to complete the proof of Lemma~\ref{lemma:correctorder}.
Consider $P_F$ and $P_H$. By Claims~\ref{cl: h is good} and~\ref{cl: f is good}, we know that they are ranges in the same sequence of variables $\var(P_F)\cup \var(P_H)$. We know that the largest variables of $P_H$ is greater than all variables of $P_f$. Then either $P_F$ is contained in $P_H$, and then we can let $P_G=P_H$ (it contains all good variables of~$G$), or we have $P_F =PQ$ and $P_H=QR$ for some words $P, Q, R$. In this case we let $P_G = P_F \circ P_H = PQQR=PQR$. Clearly, $\mmin_G$ is the suffix of $P_G$ and $P_G$ itself is the suffix of $\mmax_G$.
\end{proof}

\subsection{Reducing Non-commutative Range Queries to Commutative Range Queries} \label{sec:non-commutative_to_commutative}

In this subsection we prove Lemma~\ref{lem:intervals}.

\begin{proof}[Proof of Lemma~\ref{lem:intervals}]
We will show that any computation of ranges over $X_{sym}$ can be reconstructed without increase in the number of gates in such a way that each gate computes a range (recall, that we call this a range circuit). It is easy to see that then this circuit can be reconstructed as a circuit over $X$ each gate of which computes the same range with the variables in the increasing order. Indeed, we need to make sure that each gate computes a range in such a way that all variables are in the increasing order and this is easy to do by induction. Each gate computes a product of two ranges $a$ and $b$. If one of them is contained in the other, we simplify the circuit, since the gate just computes the same range as one of its inputs (due to idempotency and commutativity). It is impossible that $a$ and $b$ are non-intersecting and have a gap between them, since then our gate does not compute a range (in a range circuit). So, if $a$ and  $b$ are non-intersecting, then they are consecutive and we just need to multiply them in the right order. If the ranges are intersecting, we just multiply then in the right order and apply idempotency.

Thus it remains to show that each circuit for range query problem over $X_{sym}$ can be reconstructed into a range circuit. For this we will need some notation.

Suppose we have some circuit ${\cal C}$. For each gate $G$ denote by $\lef(G)$ the smallest index of the variable in $G$ (the leftmost variable). Analogously denote by $\righ(G)$ the largest index of the variable in $G$. Denote by $\gap(G)$ the smallest $i$ such that $x_i$ is not in $G$, but there are some $j,k$ such that $j<i<k$ and $x_j$ and $x_k$ (the smallest index of the variable that is in the gap in $G$).
Next, fix some topological ordering of gates in ${\cal C}$ (the ordering should be proper, that is inputs to any gate should have smaller numbers). Denote by $\num(G)$ the number of a gate in this ordering. Finally, by $\out(G)$ denote the out-degree of $G$.

For each gate that computes a non-range consider the tuple
$$
\tup(G)=(\lef(G),\gap(G),\num(G),-\out(G)).
$$ For the circuit ${\cal C}$ consider $\tup({\cal C}) = \min_G \tup(G)$, where the minimum is considered in the lexicographic order and is taken over all non-range gates. If there are no non-range gates we let $\tup({\cal C})=\infty$. This is our semi-invariant, we will show that if we have a circuits that is not a range circuit, we can reconstruct it to increase  its $\tup$ (in the lexicographic order) without increasing its size. Since $\tup$ ranges over a finite set, we can reconstruct the circuit repeatedly and end up with a range circuit.

Now we are ready to describe a reconstruction of a circuit. Consider a circuit ${\cal C}$ that is not a range circuit. And consider a gate $G$ such that $\tup(G)=\tup({\cal C})$ (it is clearly unique). Denote by $A$ and $B$ two inputs of $G$. Let $i=\lef(G)$ and $j=\gap(G)$, that is $x_i$ is the variable with the smallest index in $G$ and $x_j$ is the first gap variable of $G$ (it is not contained in $G$).

The variable $x_i$ is contained in at least one of $A$ and $B$. Consider the gate among $A$ and $B$ that contains $x_i$. This gate cannot have $x_j$ or earlier variable as a gap variable: it would contradict minimality of $G$ (by the second or the third coordinate of $\tup$). Thus this gate is a range $[x_i,x_{j'})$ for some $j'\leq j$ (by this we denote the product of variables from $x_i$ to $x_{j'}$ excluding $x_{j'}$). In particular, only one of $A$ and $B$ contains $x_i$: otherwise they are both ranges and $x_j$ is not a gap variable for $G$.

From now on we assume that $A$ contains $x_i$, that is $A=[x_i,x_{j'})$.

Now we consider all gates $H_1,\ldots, H_k$ that have edges leading from $G$. Denote by $F_1,\ldots, F_k$ their other inputs. If $k$ is equal to $0$, we can remove $G$ and reduce the circuit. Now we consider cases.

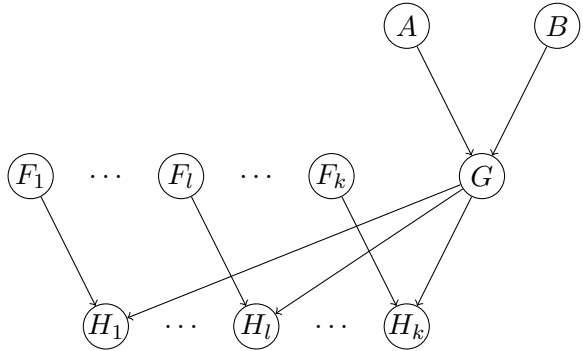
\begin{center}
\begin{tikzpicture}
\foreach \x/\y/\n/\t in {0/4/f1/F_1, 2/4/fl/F_l, 4/4/fk/F_k, 6/4/g/G, 1/2/h1/H_1, 3/2/hl/H_l, 5/2/hk/H_k, 5/6/a/A, 7/6/b/B}
  \node[inner sep=0mm,circle,draw,minimum size=6mm] (\n) at (\x,\y) {$\t$};
\foreach \x/\y/\n/\t in {1/4/fdots1/\ldots, 3/4/fdots2/\ldots, 2/2/hdots1/\ldots, 4/2/hdots2/\ldots}
  \node[inner sep=0mm] (\n) at (\x,\y) {$\t$};
\foreach \s/\t in {f1/h1, fl/hl, fk/hk, a/g, b/g, g/h1, g/hl, g/hk}
  \draw[->] (\s) -- (\t);

\end{tikzpicture}
\end{center}

\emph{Case 1.} Suppose that there is $l$ such that $\lef(F_l) \leq \lef(G)$. If $\lef(F_l) < \lef(G)$, then $F_l$ must contain all variables $x_i, \ldots, x_j$, since otherwise either $F_l$ or $H_l$ will have smaller $\tup$ then $G$. Thus $F_l$ contains $A$. Then, we can restructure the circuit by feeding $B$ to $H_l$ instead of $G$. This does not change the value of the gate computed by $H_l$ and reduces $\out(G)$. Thus $\tup({\cal C})$ increases and we are done.

If $\lef(F_l) = \lef(G)$, then $F_l$ still cannot have gap variables among $x_i, \ldots, x_{j-1}$ as it would contradict the minimality of $G$. Thus, $F_l$ is either a range, or it is not a range, but contain all variables $x_i, \ldots, x_{j-1}$. In the latter case again $F_l$ contains $A$. In the former case $F_l$ either contains $A$, or is contain in $G$. If $F_l$ contains $A$, we can again simplify the circuit as above. If $F_l$ is contained in $G$, we have $G=H_l$, so we can remove $H_l$ from the circuit and reduce the size of the circuit.

\emph{Case 2.} Suppose that for all $l$ we have $\lef(F_l)>\lef(G)$. Consider $l$ such that $F_l$ has the minimal $\righ(F_l)$ (if there are several such $l$ pick among them the one with the minimal $\num(F_l)$). For convenience of notation let $l=k$. Now we restructure the circuit in the following way. We feed $F_k$ to $G$ instead of $A$. We feed $A$ to $H_k$ instead of $F_k$. We feed $H_k$ to all other $H_p$'s instead of $G$.

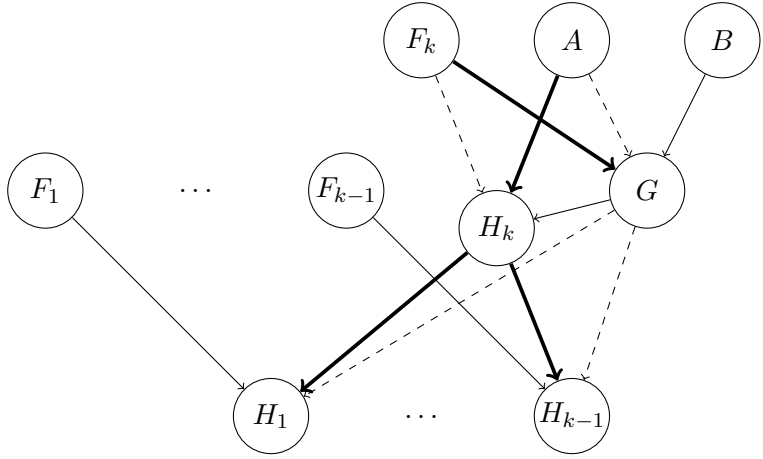
\begin{center}
\begin{tikzpicture}
\foreach \x/\y/\n/\t in {-2/4/f1/F_1, 2/4/fk1/F_{k-1}, 3/6/fk/F_k, 6/4/g/G, 1/1/h1/H_1, 5/1/hk1/H_{k-1}, 4/3.5/hk/H_k, 5/6/a/A, 7/6/b/B}
  \node[inner sep=0mm,circle,draw,minimum size=10mm] (\n) at (\x,\y) {$\t$};
\foreach \x/\y/\n/\t in {0/4/fdots1/\ldots, 3/1/hdots1/\ldots}
  \node[inner sep=0mm] (\n) at (\x,\y) {$\t$};
\foreach \s/\t in {f1/h1, fk1/hk1, b/g, g/hk}
  \draw[->] (\s) -- (\t);
\foreach \s/\t in {fk/hk, a/g, g/h1, g/hk1}
  \draw[->,dashed] (\s) -- (\t);
\foreach \s/\t in {fk/g, a/hk, hk/h1, hk/hk1}
  \draw[->,line width=0.5mm] (\s) -- (\t);

\end{tikzpicture}
\end{center}

Observe that all these reconstructions are valid, that is, they do not create directed cycles in the circuit. To verify this we need to check that there are no cycles using new edges. Indeed, there cannot be a cycle going through one of the edges $(H_k,H_p)$ since this would mean that there was a directed path from $H_p$ to one of the vertices $F_k$, $A$ and $G$ on the original circuit. Such a path to $A$ or $G$ would mean a cycle in the original circuit. Such a path to $F_k$ violates the minimality property of $F_k$ (minimal $\righ(F_k)$). Next, there cannot be a cycle going through both edges $(F_k,G)$ and $(A,H_k)$, since substituting these edges by $(F_k,H_k)$ and $(A,G)$ we obtain one or two cycles in the original circuit. Next, there cannot be a cycle going through the edge $(A,H_k)$ only, since $H_k$ is reachable from $A$ in the original circuit and this would mean a cycle in the original circuit. Finally, there cannot be a cycle going only through the edge $(F_k,G)$ since this would mean a directed path from $G$ to $F_k$ in the original circuit and this contradicts $\lef(F_k)>\lef(G)$.

Note that our reconstruction might require reordering of the circuit gates, since we create edges between previously incomparable $H$-gates and between $F_k$ and $G$. But the reordering affect only the gates with $\num$ greater than $\num(G)$ and may only reduce $\num(F_k)$ to be smaller than $\num(G)$. But this can only increase $\tup(G)$ and since $\lef(F_k)>\lef(G)$ this can only increase $\tup({\cal C})$.

Observe, that the circuit still computes the outputs correctly. The changes are in the gates $H_1\ldots, H_k$ (and also in $G$, but $H_1,\ldots, H_k$ are all of its outputs). $H_k$ does not change. Other $H_p$'s might have changed, they now additionally include variables of $F_k$. But note that all of these variables are in between of $\lef(H_p)$ and $\righ(H_p)$, so they must be presented in the output gates connected to $H_p$ anyway (recall that at the output gates we compute ranges).

Now, observe that $\tup(G)$ has increased (by the first coordinate). There are no new gates with smaller $\lef$. Among gates with the minimal $\lef$ there are no new gates with smaller $\gap$. Among gates with minimal $(\lef,\gap)$ all gates have larger $\num$ then $G$. Thus $\tup({\cal C})$ increased and we are done.
\end{proof}

\section{Open Problems}
There are two natural problems left open.
\begin{enumerate}
\item Design a~deterministic $O(z)$ time algorithm for generating
a~circuit in the commutative case.
For this, it suffices to design an $O(n)$ deterministic algorithm for the following problem: given a~list of positions
of $n$~zeroes of an $n \times n$ 0/1-matrix with at most $\log n$ zeroes in every row, permute its columns so that the total length of all segments of length at most $O(\log n)$ is $O(\frac{n}{\log n})$.
\item Determine the asymptotic complexity of the linear operator in terms of the number of zeroes in the non-commutative case.
\end{enumerate}

\section*{Acknowledgments}
We thank Paweł Gawrychowski for pointing us out to the paper~\cite{DBLP:journals/ijcga/ChazelleR91}. We thank Alexey Talambutsa for fruitful discussions on the theory of semigroups.

\bibliographystyle{plain}
\bibliography{text}

\clearpage
\appendix
\section{Review}
\subsection{Algebraic Structures}\label{subsec:algstr}

A~\emph{semigroup} $(S, \circ)$ is an algebraic structure, where the operation
$\circ$ is \emph{closed}, i.e., $\circ : S\times S \rightarrow S$, and
\emph{associative}, i.e.,
$x \circ (y \circ z) = (x \circ y) \circ z$ for all $x$, $y$, and $z$ in $S$.
\emph{Commutative} (or \emph{abelian}) semigroups introduce one extra
requirement: $x \circ y = y \circ x$ for all $x$ and $y$ in $S$.

Commutative semigroups are ubiquitous. Below we list a few
notable examples, starting with the most basic one, which is, arguably, known
to every person on the planet.


\begin{itemize}
  \item Integer numbers form commutative semigroups with many operations. For
  example, the order in which numbers are \emph{added} is irrelevant, hence
  $(\mathbb{Z}, +)$ is a commutative semigroup. So are $(\mathbb{Z}, \times)$,
  $(\mathbb{Z}, \min)$ and $(\mathbb{Z}, \max)$. On the other hand, it does
  matter in which order numbers are \emph{subtracted}, hence $(\mathbb{Z}, -)$
  is not a commutative semigroup: $1-2 \neq 2-1$. In fact, $(\mathbb{Z}, -)$
  is not even a semigroup, since subtraction is non-associative:
  $1-(2-3) \neq (1-2)-3$.

  \item Boolean values form commutative semigroups $(\mathbb{B}, \vee)$,
  $(\mathbb{B}, \wedge)$, $(\mathbb{B}, \oplus)$ and $(\mathbb{B}, \equiv)$.

  \item Any commutative semigroup $(S, \circ)$ can be \emph{lifted} to the set
  $\hat{S}$ of ``containers'' of elements $S$, e.g., vectors or matrices,
  obtaining a commutative semigroup $(\hat{S}, \hat{\circ})$, where the lifted
  operation~$\hat{\circ}$ is applied to the contents of containers element-wise.
  The lifting operation $\hat{\cdot}$ can often be omitted for clarity if there
  is no ambiguity.

  The \emph{average semigroup} $(\mathbb{Z} \times \mathbb{Z}, \circ)$ is a
  simple yet not entirely trivial example of semigroup lifting. By defining
  $(t_1, c_1) \circ (t_2, c_2) = (t_1 + t_2, c_1 + c_2)$, we can aggregate
  partial \emph{totals} and \emph{counts} of a set of numbers, which allows us
  to efficiently calculate their average as $\textit{avg}(t, c) = \frac{t}{c}$.
  The average semigroup is commutative.

  \item Set union and intersection are commutative and associative operations
  giving rise to many set-based commutative semigroups. Here we highlight an
  example that motivated our research: the \emph{graph overlay} operation,
  defined\footnote{This definition coincides with that of the
  \emph{graph union} operation~\cite{1969_graph_theory_harary}. Graph union
  typically requires that given graphs are non-overlapping, hence it is not
  closed on the set of all graphs. Graph overlay does not have such a
  requirement, and is therefore closed and forms a semigroup.} as
  $(V_1, E_1) + (V_2, E_2) = (V_1 \cup V_2, E_1 \cup E_2)$,~where~$(V, E)$ is
  a standard set-based representation for directed unweighted graphs, comes from
  an algebra of graphs used in functional programming~\cite{mokhov2017algebraic}.
  See further details in Section~\ref{sec-dense-graph}.
\end{itemize}

\emph{Groups} extend semigroups by requiring the existence of the \emph{identity
element} $0 \in S$, such that $0 \circ x = x \circ 0=x$, and the \emph{inverse
element} $-x \in S$ for all $x \in S$, such that
$(-x) \circ x = x \circ (-x) = 0$. Groups provide a natural generalisation of
arithmetic \emph{subtraction}, whereby $x \circ (-y)$ denotes the subtraction of
$y$ from $x$.

A commutative semigroup $(S, \circ)$ can often be extended to a \emph{semiring}
$(S, \circ, \bullet)$ by introducing another associative (but not necessarily
commutative) operation $\bullet$ that \emph{distributes} over $\circ$, that is
\[
x \bullet (y \circ z) = (x \bullet y) \circ (x \bullet z)\\
\]
\[
(x \circ y) \bullet z = (x \bullet z) \circ (y \bullet z)
\]
hold for all $x$, $y$, and $z$ in~$S$. Since $\circ$ and $\bullet$~behave
similarly to numeric addition and multiplication, it is common to give $\bullet$
a higher precedence to avoid unnecessary parentheses, and even omit~$\bullet$~from
formulas altogether, replacing it by juxtaposition. This gives a terser and
more convenient notation, e.g., the left distributivity law becomes:
$x (y \circ z) = x y \circ x z$. We will use this notation, insofar as this does
not lead to ambiguity.

Most definitions of semirings also require that the two operations have
identities: the \emph{additive identity}, denoted by 0, such that
$0 \circ x = x \circ 0=x$, and the \emph{multiplicative identity}, denoted by 1,
such that $1x=x1=x$. Furthermore, 0 is typically required to be
\emph{annihilating}: $0x=x0=0$.

A semiring $(S, \circ, \bullet)$ is also a \emph{ring} if $(S, \circ)$ is a
group, i.e., the operation $\circ$ is invertible. One can think of rings as
semirings with subtraction.

Let us extend some of our semigroup examples to semirings:

\begin{itemize}
  \item The most basic and widely known semiring is that of integer numbers with
  addition and multiplication: $(\mathbb{Z}, +, \times)$. Since every integer
  number $x\in \mathbb{Z}$ has an inverse $-x \in \mathbb{Z}$ with respect to
  the addition operation, $(\mathbb{Z}, +, \times)$ is also a ring.
  Interestingly, integer addition can also play the role of multiplication when
  combined with the $\max$ operation, resulting in the \emph{tropical semiring}
  $(\mathbb{Z}, \max, +)$, which is also known as the \emph{max-plus algebra}.
  Unlike $+$, the $\max$ operation has no inverse, therefore
  $(\mathbb{Z}, \max, +)$ is not a ring.

  \item Boolean values form the semiring $(\mathbb{B}, \vee, \wedge)$. Note that
  $(\mathbb{B}, \wedge, \vee)$ is a semiring too thanks to the duality between
  the operations $\vee$ and $\wedge$. Furthermore,
  $(\mathbb{B}, \oplus, \wedge)$ is a ring, where every element is its own
  inverse: $x \oplus x = 0$ for $x \in \mathbb{B}$.

  \item Semirings and rings $(S, \circ, \bullet)$ can also be lifted to the set
  $\hat{S}$ of ``containers'' of elements $S$, most commonly matrices, obtaining
  $(\hat{S}, \hat{\circ}, \hat{\bullet})$. Matrices over tropical semirings, for
  example, are used for solving various path-finding problems on graphs.
\end{itemize}


\subsection{Applications of the Range Queries Problem}\label{subseq:rmqapp}
There are many natural applications of the range queries problem for a~collection of records in a~database: computing the total population of cities that are at most some distance away from a~given point, computing an average salary in a~given period of time, finding the minimum depth on a~given subrectangle on a~sea map, etc. Below, we review some of the less straightforward applications where efficient algorithms for the range queries problem are usually combined with other algorithmic ideas.
\begin{description}
\item[String algorithms and computational biology.]
It is possible to preprocess a~given string in $O(n)$ time (where $n$ is its length) so that to then find the longest common prefix of any two suffixes of the original string in constant time. This is done by first constructing the suffix array and the longest common prefix array of the string and then using an efficient RMQ algorithm.

\item[Computational geometry.] Algorithms for the range queries problems can be used together with the scanning line technique to solve efficiently various problems like: given a~set of segments on a~line, compute the number of intersecting pairs of segments; or, given a~set of rectangles and a~set of points on a~plane, compute, for each each rectangle, the number of points it contains. 


\end{description}


\subsection{Known Approaches to Range Queries}\label{subsec:approaches}
In this subsection, we give a~brief overview of a~rich variety of known algorithms for the range queries problem. We say that an algorithm has type $(f(n), g(n))$ if it spends $f(n)$ time on preprocessing the input sequence, and then answers any query in time $g(n)$.

\begin{description}
\item[No preprocessing.] A~naive algorithm skips the preprocessing stage and answers a~query $(l,r)$ directly in time $O(r-l+1)$. It therefore has type $(O(1), O(n))$.

\item[Full preprocessing.] One may precompute the answers to all possible queries to be able to answer any subsequent query immediately. Using dynamic programming, it is possible to precompute the answers to all $\Theta(n^2)$ queries in time $O(n^2)$: for this, it is enough to process the queries in order of increasing length. This gives an $(O(n^2), O(1))$ algorithm.

\item[Fixed length queries (sliding window).] In case one is promised that all the queries are going to have the same length~$m$, it is possible to do an~$O(n)$ time preprocessing and then to answer any query in time $O(1)$. For this, one partitions the input sequence of size~$n$ into $\frac nm$ blocks of size~$m$. For each block, one computes all its prefixes and suffixes in time $O(m)$. The overall running time is $O(\frac nm \cdot m)=O(n)$. Then, each query of length~$m$ touches at most two consecutive blocks and can be answered by taking a~precomputed suffix of the left block and a~precomputed prefix of the right block in time $O(1)$. This, in particular, implies that, given a~sequence of length~$n$ and an integer $1 \le m \le n$, one may slide a~window of length~$m$ through the sequence and to output the answer to all such window queries in time $O(n)$.

\item[Prefix sums.] In case, a~semigroup operation has an {\em easily computable inverse}, it is easy to design an $(O(n), O(1))$ algorithm. We illustrate this for a~group $(\mathbb{Z}, +)$. Given $x_1, \dotsc, x_n$, we compute $(n+1)$ prefix sums:
\(S_0=0,\, S_1=x_1,\, S_2=x_1+x_2, \dotsc, S_n=x_1+\dotsb+x_n\,.\)
This can be done in time $O(n)$ since $S_i=S_{i-1}+x_i$. Then, the answer to any query $(l,r)$ is just $S_r-S_{l-1}$.

Note that the algorithm above solves a~{\em static} version of the problem. For the {\em dynamic} version, where one is allowed to change the elements of the input sequence, there is a~data structure known as Fenwick's tree~\cite{DBLP:journals/spe/Fenwick94}. It allows to change any element as well as to retrieve any prefix sum in time $O(\log n)$.

\item[Block decomposition.] One decomposes the input range $(1,n)$ into $n/b$~blocks of length~$b$ and then computes, for each block, all its prefixes and suffixes. This can be done in time $O(n)$. Then, for each query, if it lies entirely in a~block, we compute the answer directly (hence, in time at most $O(b)$). If it crosses a~number of blocks, we decompose it into a~suffix of a~block, a~number of consecutive blocks, and a~prefix of a~block. This allows us to answer such long queries in time $O(n/b)$. Setting $b=\sqrt{n}$ to balance both cases, we get a~$(O(n), O(\sqrt{n}))$-algorithm. 

\item[Sparse table.] This data structure works for idempotent semigroups ({\em bands}) and has type $(O(n\log n), O(1))$. We illustrate its main idea for the {\em range minimum query} (RMQ)  (i.e., for a~semigroup $(\mathbb{Z}, \min)$). One precomputes answers to $O(n\log n)$ queries---namely, those whose length is a~power of~2. More formally, for all $0 \le k \le \log_2n$ and $1 \le i \le n-2^k+1$, let $S_{k,i}$ be the answer to a~query $(i, i+2^k-1)$:
\(S_{k,i}=x_i \circ x_{i+1} \circ \dotsb \circ x_{i+2^k-1} \, .\)
Since any range of length $2^k$ consists of two ranges of length $2^{k-1}$, one can compute all $S_{k,i}$'s in time $O(n\log n)$ using dynamic programming. Then, any range $(l,r)$ can be covered by two precomputed ranges: if $k$~is the smallest integer such that $2^k \ge (r-l+1)/2$, then the answer to this query is $S_{k,l} \circ S_{k,r-2^k+1}$ (idempotency is required since we are covering the range, but not partitioning it). This gives an $(O(n\log n), O(1))$ algorithm.

\item[Hybrid strategy.] One may extend the block decomposition 
approach further and use one efficient data structure on top of 
blocks and possibly a~different data structure for each block. 
Namely, we decompose the input range into blocks of size~$b$, 
use a~$(p_1(n), q_1(n))$-algorithm on top of blocks and
a~$(p_2(n), q_2(n))$-algorithm within each block. The resulting algorithm then has type
\[(O(n + p_1(n / b) + (n / b) \cdot p_2(b), O(q_1(n/b) + q_2(b))) \, .\]
E.g., for the range minimum problem, combining the sparse table data structure ($p_1(n)=O(n\log n)$, $q_1(n)=O(1)$) with no preprocessing technique ($p_2(n)=O(1)$, $q_2=O(n)$) and block size $b=\log n$, gives an~$(O(n), O(\log n))$-algorithm. Another example: using sparse table in both cases (with block size $b=\log n$) gives an $(O(n\log\log n), O(1))$ algorithm. 


\item[Segment tree.] The segment tree data structure is also based on dynamic programming ideas and works for any semigroup. Consider the following complete binary tree with $O(n)$ nodes: the root is labeled by a~query $(1,n)$, the two children of each inner node $(l,r)$ are labeled by the left and right halves of the current query (i.e., $(l,m)$ and $(m+1,r)$ where $m=(l+r)/2$), the leaves are labeled by length one queries. Going from leaves to the root, one can precompute the answers to all the queries in this tree in time $O(n)$. Then, it is possible to show that any query $(l,r)$ can be  partitioned into $O(\log n)$ queries that are stored in the tree. This gives an $(O(n), O(\log n))$ algorithm. It should be noted that the segment tree can also be used to solve the dynamic version of the range queries problem efficiently: to change the value of one of the elements of the input sequence, one needs to adjust the answers to $O(\log n)$ queries stored in the tree.

\item[Algorithms by Yao and by Alon and Schieber.] Yao~\cite{DBLP:conf/stoc/Yao82} showed that, for any semigroup, it is possible to preprocess the input sequence in time $O(n)$ so that to further answer any query in time $O(\alpha(n))$ where $\alpha(n)$ is the inverse Ackermann function and proved a~matching lower bound. Later, Alon and Schieber~\cite{Alon87optimalpreprocessing} studied a~more specific question: what is the minimum number of semigroup operations needed at the preprocessing stage for being able to then answer any query in at most $k$~steps? They proved matching lower and upper bounds for every~$k$. As a~special case, they show how to preprocess the input sequence in time $O(n\log n)$ so that to answer any subsequent query by applying at most one semigroup operation. This algorithm generalizes the sparse table data structure (as it does not require the semigroup to be idempotent) and is particularly easy to describe. It is based on the divide-and-conquer paradigm. Let $m=n/2$. We precompute answers to all queries of the form $(i,m)$ and $(m+1,j)$, where $1 \le i \le m$ and $m+1 \le j \le n$ (i.e., suffixes of the left half and prefixes of the right half). This allows to answer in a~single step any query that intersects the middle of the sequence, i.e., queries $(l,r)$ such that $l \le m \le r$. All the remaining preprocessing boils down to answering queries that lie entirely in either left or right half. This can be done recursively for the halves. The corresponding recurrence relation $T(n)=2T(n/2)+O(n)$ implies an upper bound $O(n\log n)$ on preprocessing time (and hence, the number of semigroup operations).

\item[$(O(n), O(1))$-type algorithms.]
There is a~sequence of $(O(n), O(1))$-type algorithms designed specifically to the range minimum query problem and a~related problem called least common ancestor (LCA) \cite{DBLP:journals/siamcomp/BerkmanV93,
DBLP:journals/jal/BenderFPSS05,
DBLP:conf/latin/BenderF00, 
DBLP:conf/cpm/FischerH06}. Here, we briefly sketch the algorithm by Bender and Farach-Colton. 
Its main idea is to first reduce RMQ to LCA (the least common ancestor problem). One then reduces LCA back to RMQ and notices that the resulting instance of RMQ has a~convenient property: the difference between any two consecutive elements is $\pm 1$. This property allows to do the following trick: we precompute answers to all relatively short queries (this can be done even without knowing the input sequence because of the $\pm 1$ property); we also partition the input sequence into blocks and build a~segment tree out of these blocks.
\end{description}


\subsection{Dense Graph Representation}\label{sec-dense-graph}

Let us revisit the graph semigroup defined in Section~\ref{subsec:algstr}.
We will denote it by $(G_U, +)$, where $G_U$ is the set of directed graphs whose
vertices come from a universe $U$, that is, if $(V, E) \in G_U$ then
$V \subseteq U$ and $E \subseteq V \times V$. Recall that the graph overlay
operation $+$ is defined as

\[
(V_1, E_1) + (V_2, E_2) = (V_1 \cup V_2, E_1 \cup E_2).
\]

\noindent
The algebra of graphs presented in~\cite{mokhov2017algebraic} also defines
the \emph{graph connect} operation $\rightarrow$:

\[
(V_1, E_1) \rightarrow (V_2, E_2) = (V_1 \cup V_2, E_1 \cup E_2 \cup V_1 \times V_2).
\]

This operation allows us to ``connect'' two graphs by adding edges from every
vertex in the left-hand graph to every vertex in the right-hand graph, possibly
introducing self-loops if $V_1 \cap V_2 \neq \emptyset$. The operation is
associative, non-commutative and distributes over $+$. Note, however, that the
empty graph $\varepsilon = (\emptyset, \emptyset)$ is the identity for both
overlay and connect operations: $\varepsilon + x = x + \varepsilon = x$ and
$\varepsilon \rightarrow x = x \rightarrow \varepsilon = x$, and consequently
the annihilating zero property does not hold, which makes this algebraic
structure not a~semiring according to the classic semiring definition.

By using the two operations one can construct any graph starting from primitive
single-vertex graphs. For example, let $U=\{1,2,3\}$ and $k$ stand for a
single-vertex graph $({k}, \emptyset)$. Then:

\begin{itemize}
  \item $1 \rightarrow 2$ is the graph comprising a single edge $(1,2)$, i.e.
  $1 \rightarrow 2 = (\{1,2\}, \{(1,2)\})$.
  \item $1 \rightarrow (2 + 3)$ is the graph $(\{1,2,3\}, \{(1,2),(1,3)\})$.
  \item $1 \rightarrow 2 \rightarrow 3$ is the graph $(\{1,2,3\}, \{(1,2),(1,3),(2,3)\})$.
\end{itemize}

\noindent
Clearly any sparse graph $(V, E)$, i.e. a graph with a sparse connectivity
matrix, can be constructed by a linear-size expression:

\[
(V, E) = \sum_{v \in V} v + \sum_{(u,v) \in E} u \rightarrow v.
\]

\noindent
But what about complements of sparse graphs, i.e. graphs with dense
connectivity matrices? It is not difficult to show that by applying the dense
linear operator we can obtain a linear-size circuit comprising 2-input gates
$+$ and $\rightarrow$ for any dense graph.

Let $A$ be a dense matrix of size $n\times n$. Our goal is to construct the
graph $G_A = (\{1, \dots, n\}, E)$ such that $(i,j) \in E$ whenever $A_{ij}=1$.

First, we compute the dense linear operator $\mathbf{y} = A \mathbf{x}$ over the
(commutative) graph semigroup~$+$, where $\mathbf{x} = (1, 2, \dots, n)$, i.e.,
$x_j$ is the primitive graph comprising a single vertex~$j$, obtaining
graphs~$y_i$ that comprise sets of isolated vertices corresponding to the rows
of matrix~$A$:

\[
y_i = \sum_{A_{ij}=1} j \, .
\]

The resulting graph $G_A = (\{1, \dots, n\}, E)$ can now be obtained by using
the connect operation~$\rightarrow$ to connect $i$ to all vertices $y_i$, and
subsequently overlaying the results:

\[
G_A = \sum_{i=1}^{n} i \rightarrow y_i.
\]

Thanks to the linear-size construction for the dense linear operator, the size
of the circuit computing $G_A$ is $O(n)$. This allows us to store dense graphs
on $n$ vertices using $O(n)$ memory, and perform basic transformations of dense
graphs in $O(n)$ time. We refer the reader to~\cite{mokhov2017algebraic} for
further details about applications of algebraic graphs in functional programming
languages.

\end{document}